\documentclass[fleqn]{llncs}
\usepackage{version}
\usepackage{hp}

\title{Autosolvability of Halting Problem Instances \\
       for Instruction Sequences}
\author{J.A. Bergstra \and C.A. Middelburg }
\institute{Informatics Institute, Faculty of Science,
           University of Amsterdam, \\
           Science Park~107, 1098~XG Amsterdam, the Netherlands \\
           \email{J.A.Bergstra@uva.nl,C.A.Middelburg@uva.nl}}

\begin{document}
\maketitle

\begin{abstract}
We position Turing's result regarding the undecidability of the halting
problem as a result about programs rather than machines.
The mere requirement that a program of a certain kind must solve the
halting problem for all programs of that kind leads to a contradiction
in the case of a recent unsolvability result regarding the halting
problem for programs.
In this paper, we investigate this autosolvability requirement in a
setting in which programs take the form of instruction sequences.
\begin{keywords}
halting problem, instruction sequence, autosolvability,
functional unit.
\end{keywords}
\begin{classcode}
F.1.1, F.4.1.
\end{classcode}
\end{abstract}

\section{Introduction}
\label{sect-intro}

The halting problem is frequently paraphrased as follows: the halting
problem is the problem to determine, given a program and an input to the
program, whether execution of the program on that input will eventually
terminate.
To indicate that this problem might be undecidable, it is often
mentioned that an interpreter, which is a program that simulates the
execution of programs that it is given as input, cannot solve the
halting problem because the interpreter will not terminate if its input
program does not terminate.
However, Turing's result regarding the undecidability of the halting
problem is a result about Turing machines rather than programs.
It says that there does not exist a single Turing machine that, given
the description of an arbitrary Turing machine and input, will determine
whether the computation of that Turing machine applied to that input
eventually halts (see e.g.~\cite{Tur37a}).

Our objective is to position Turing's result regarding the
undecidability of the halting problem as a result about programs rather
than machines.
In the case of the unsolvability result regarding the halting problem
for programs presented in~\cite{BP04a}, the mere requirement that a
program of a certain kind must solve the halting problem for all
programs of that kind leads to a contradiction.
In this paper, we pay closer attention to this autosolvability
requirement.
Like in~\cite{BP04a}, we carry out our investigation in a setting in
which programs take the form of instruction sequences.
The instruction set concerned includes instructions whose processing
needs a device that resembles the tape of a Turing machine.

The work presented in this paper belongs to a line of research in which
program algebra~\cite{BL02a} is the setting used for investigating
issues in which instruction sequences are involved.
The starting-point of program algebra is the perception of a program as
a single-pass instruction sequence, i.e.\ a finite or infinite sequence
of instructions of which each instruction is executed at most once and
can be dropped after it has been executed or jumped over.
Moreover, basic thread algebra~\cite{BL02a} is the setting used for
modelling the behaviours exhibited by instruction sequences under
execution.%
\footnote
{In~\cite{BL02a}, basic thread algebra is introduced under the name
 basic polarized process algebra.
}
In this paper, we use a program notation rooted in program algebra,
instead of program algebra itself.
The program notation in question was first presented in~\cite{BM09l}.
In that paper, the concept of a functional unit is introduced and
studied.
Here, we will model the devices that resemble the tape of a Turing
machine by a functional unit.

This paper is organized as follows.
First, we give a survey of the program notation used in this paper
(Section~\ref{sect-PGLBbt}) and define its semantics using basic thread
algebra (Section~\ref{sect-BTAbt}).
Next, we extend basic thread algebra with operators that are related to
the processing of instructions by services (Section~\ref{sect-TAbt}) and
introduce those operator in the setting of the program notation used
(Section~\ref{sect-PGLBbt-add}).
Then, we introduce the concept of a functional unit
(Section~\ref{sect-func-unit}) and define autosolvability and related
notions in terms of functional units related to Turing machine tapes
(Section~\ref{sect-func-unit-sbs}).
After that, we discuss the weakness of interpreters when it comes to
solving the halting problem (Section~\ref{sect-interpreters}) and give
positive and negative results concerning the autosolvability of the
halting problem (Section~\ref{sect-autosolvability}).
Finally, we make some concluding remarks (Section~\ref{sect-concl}).

\section{\PGLB\ with Boolean Termination}
\label{sect-PGLBbt}

In this section, we give a survey of the program notation \PGLBbt.
This program notation is a variant of the program notation \PGLB, which
belongs to a hierarchy of program notations rooted in program algebra
(see~\cite{BL02a}).
\PGLBbt\ is \PGLB\ with the Boolean termination instructions $\haltP$
and $\haltN$ from~\cite{BM09k} instead of the termination instruction
$\halt$.
\PGLB\ and \PGLBbt\ are close to existing assembly languages and have
relative jump instructions.

In \PGLBbt, it is assumed that fixed but arbitrary non-empty finite set
$\BInstr$ of \emph{basic instructions} has been given.
The intuition is that the execution of a basic instruction may modify a
state and produces $\True$ or $\False$ at its completion.

\PGLBbt\ has the following primitive instructions:
\begin{itemize}
\item
for each $a \in \BInstr$, a \emph{plain basic instruction} $a$;
\item
for each $a \in \BInstr$, a \emph{positive test instruction} $\ptst{a}$;
\item
for each $a \in \BInstr$, a \emph{negative test instruction} $\ntst{a}$;
\item
for each $l \in \Nat$, a \emph{forward jump instruction}
$\fjmp{l}$;
\item
for each $l \in \Nat$, a \emph{backward jump instruction}
$\bjmp{l}$;
\item
a \emph{positive termination instruction} $\haltP$;
\item
a \emph{negative termination instruction} $\haltN$.
\end{itemize}
\PGLBbt\ instruction sequences have the form $u_1 \conc \ldots \conc
u_k$, where $u_1,\ldots,u_k$ are primitive instructions of \PGLBbt.

On execution of a \PGLBbt\ instruction sequence, these primitive
instructions have the following effects:
\begin{itemize}
\item
the effect of a positive test instruction $\ptst{a}$ is that basic
instruction $a$ is executed and execution proceeds with the next
primitive instruction if $\True$ is produced and otherwise the next
primitive instruction is skipped and execution proceeds with the
primitive instruction following the skipped one -- if there is no
primitive instructions to proceed with, deadlock occurs;
\item
the effect of a negative test instruction $\ntst{a}$ is the same as the
effect of $\ptst{a}$, but with the role of the value produced reversed;
\item
the effect of a plain basic instruction $a$ is the same as the effect of
$\ptst{a}$,\linebreak[2] but execution always proceeds as if $\True$ is
produced;
\item
the effect of a forward jump instruction $\fjmp{l}$ is that execution
proceeds with the $l$-th next primitive instruction -- if $l$ equals $0$
or there is no primitive instructions to proceed with, deadlock occurs;
\item
the effect of a backward jump instruction $\bjmp{l}$ is that execution
proceeds with\linebreak[2] the $l$-th previous primitive instruction --
if $l$ equals $0$ or there is no primitive\linebreak[2] instructions to
proceed with, deadlock occurs;
\item
the effect of the positive termination instruction $\haltP$ is that
execution terminates and in doing so delivers the Boolean value $\True$;
\item
the effect of the negative termination instruction $\haltP$ is that
execution terminates and in doing so delivers the Boolean value
$\False$.
\end{itemize}

\section{Thread Extraction}
\label{sect-BTAbt}

In this section, we make precise in the setting of \BTAbt\ (Basic Thread
Algebra with Boolean termination) which behaviours are exhibited on
execution by \PGLBbt\ instruction sequences.
We start by reviewing \BTAbt.

In \BTAbt, it is assumed that a fixed but arbitrary non-empty finite set
$\BAct$ of \emph{basic actions}, with $\Tau \not\in \BAct$, has been
given.
We write $\BActTau$ for $\BAct \union \set{\Tau}$.
The members of $\BActTau$ are referred to as \emph{actions}.

A thread is a behaviour which consists of performing actions in a
sequential fashion.
Upon each basic action performed, a reply from an execution environment
determines how it proceeds.
The possible replies are the Boolean values $\True$ (standing for true)
and $\False$ (standing for false).
Performing the action $\Tau$ will always lead to the reply $\True$.

\BTAbt\ has one sort: the sort $\Thr$ of \emph{threads}.
We make this sort explicit because we will extend \BTAbt\ with
additional sorts in Section~\ref{sect-TAbt}.
To build terms of sort $\Thr$, \BTAbt\ has the following constants and
operators:
\begin{itemize}
\item
the \emph{deadlock} constant $\const{\DeadEnd}{\Thr}$;
\item
the \emph{positive termination} constant $\const{\StopP}{\Thr}$;
\item
the \emph{negative termination} constant $\const{\StopN}{\Thr}$;
\item
for each $a \in \BActTau$, the binary \emph{postconditional composition}
operator $\funct{\pcc{\ph}{a}{\ph}}{\linebreak[2]\Thr \x \Thr}{\Thr}$.
\end{itemize}
We assume that there is a countably infinite set of variables of sort
$\Thr$ which includes $x,y,z$.
Terms of sort $\Thr$ are built as usual.
We use infix notation for postconditional composition.
We introduce \emph{action prefixing} as an abbreviation: $a \bapf p$,
where $p$ is a term of sort $\Thr$, abbreviates $\pcc{p}{a}{p}$.

The thread denoted by a closed term of the form $\pcc{p}{a}{q}$ will
first perform $a$, and then proceed as the thread denoted by $p$
if the reply from the execution environment is $\True$ and proceed as
the thread denoted by $q$ if the reply from the execution environment is
$\False$.
The threads denoted by $\DeadEnd$, $\StopP$ and $\StopN$ will become
inactive, terminate with Boolean value $\True$ and terminate with
Boolean value $\False$, respectively.

\BTAbt\ has only one axiom.
This axiom is given in Table~\ref{axioms-BTAbt}.%
\begin{table}[!tb]
\caption{Axiom of \BTAbt}
\label{axioms-BTAbt}
\begin{eqntbl}
\begin{axcol}
\pcc{x}{\Tau}{y} = \pcc{x}{\Tau}{x}                     & \axiom{T1}
\end{axcol}
\end{eqntbl}
\end{table}

Each closed \BTAbt\ term of sort $\Thr$ denotes a thread that will
become inactive or terminate after it has performed finitely many
actions.
Infinite threads can be described by guarded recursion.
A \emph{guarded recursive specification} over \BTAbt\ is a set of
recursion equations $E = \set{x = t_x \where x \in V}$, where $V$ is a
set of variables of sort $\Thr$ and each $t_x$ is a \BTAbt\ term of the
form $\DeadEnd$, $\StopP$, $\StopN$ or $\pcc{t}{a}{t'}$ with $t$ and
$t'$ that contain only variables from $V$.
We are only interested in models of \BTAbt\ in which guarded recursive
specifications have unique solutions.
Regular threads, i.e.\ threads that can only be in a finite number of
states, are solutions of finite guarded recursive specifications.

To reason about infinite threads, we assume the infinitary conditional
equation AIP (Approximation Induction Principle).
AIP is based on the view that two threads are identical if their
approximations up to any finite depth are identical.
The approximation up to depth $n$ of a thread is obtained by cutting it
off after it has performed $n$ actions.
In AIP, the approximation up to depth $n$ is phrased in terms of the
unary \emph{projection} operator $\funct{\projop{n}}{\Thr}{\Thr}$.
AIP and the axioms for the projection operators are given in
Table~\ref{axioms-AIP}.%
\begin{table}[!tb]
\caption{Approximation induction principle}
\label{axioms-AIP}
\begin{eqntbl}
\begin{axcol}
\AND{n \geq 0}{} \proj{n}{x} = \proj{n}{y} \Implies
                                                  x = y & \axiom{AIP} \\
\proj{0}{x} = \DeadEnd                                  & \axiom{P0} \\
\proj{n+1}{\StopP} = \StopP                             & \axiom{P1a} \\
\proj{n+1}{\StopN} = \StopN                             & \axiom{P1b} \\
\proj{n+1}{\DeadEnd} = \DeadEnd                         & \axiom{P2} \\
\proj{n+1}{\pcc{x}{a}{y}} =
                      \pcc{\proj{n}{x}}{a}{\proj{n}{y}} & \axiom{P3}
\end{axcol}
\end{eqntbl}
\end{table}
In this table, $a$ stands for an arbitrary action from $\BActTau$ and
$n$ stands for an arbitrary natural number.

We can prove that the projections of solutions of guarded recursive
specifications over \BTAbt\ are representable by closed \BTAbt\ terms of
sort $\Thr$.
\begin{lemma}
\label{lemma-projections}
Let $E$ be a guarded recursive specification over \BTAbt, and
let $x$ be a variable occurring in $E$.
Then, for all $n \in \Nat$, there exists a closed \BTAbt\ term $p$ of
sort $\Thr$ such that $E \Implies \proj{n}{x} = p$.
\end{lemma}
\begin{proof}
In the case of \BTA, this is proved in~\cite{BM06a} as part of the proof
of Theorem~1 from that paper.
The proof concerned goes through in the case of \BTAbt.
\qed
\end{proof}

The behaviours exhibited on execution by \PGLBbt\ instruction sequences
are considered to be regular threads, with the basic instructions taken
for basic actions.
The \emph{thread extraction} operation $\extr{\ph}$ defines, for each
\PGLBbt\ instruction sequence, the behaviour exhibited on execution by
that \PGLBbt\ instruction sequence.
The thread extraction operation is defined by
$\extr{u_1 \conc \ldots \conc u_k} =
 \extr{1,u_1 \conc \ldots \conc u_k}$,
where $\extr{\ph,\ph}$ is defined by the equations given in
Table~\ref{axioms-thread-extr} (for $a \in \BInstr$ and $l,i \in \Nat$)%
\begin{table}[!t]
\caption{Defining equations for thread extraction operation}
\label{axioms-thread-extr}
\begin{eqntbl}
\begin{aceqns}
\extr{i,u_1 \conc \ldots \conc u_k} & = & \DeadEnd
& \mif \mathrm{not}\; 1 \leq i \leq k \\
\extr{i,u_1 \conc \ldots \conc u_k} & = &
a \bapf \extr{i+1,u_1 \conc \ldots \conc u_k}
& \mif u_i = a \\
\extr{i,u_1 \conc \ldots \conc u_k} & = &
\pcc{\extr{i+1,u_1 \conc \ldots \conc u_k}}{a}
    {\extr{i+2,u_1 \conc \ldots \conc u_k}}
& \mif u_i = +a \\
\extr{i,u_1 \conc \ldots \conc u_k} & = &
\pcc{\extr{i+2,u_1 \conc \ldots \conc u_k}}{a}
    {\extr{i+1,u_1 \conc \ldots \conc u_k}}
& \mif u_i = -a \\
\extr{i,u_1 \conc \ldots \conc u_k} & = &
\extr{i+l,u_1 \conc \ldots \conc u_k}
& \mif u_i = \fjmp{l} \\
\extr{i,u_1 \conc \ldots \conc u_k} & = &
\extr{i \monus l,u_1 \conc \ldots \conc u_k}
& \mif u_i = \bjmp{l} \\
\extr{i,u_1 \conc \ldots \conc u_k} & = & \StopP
& \mif u_i = \haltP \\
\extr{i,u_1 \conc \ldots \conc u_k} & = & \StopN
& \mif u_i = \haltN
\end{aceqns}
\end{eqntbl}
\end{table}
and the rule that $\extr{i,u_1 \conc \ldots \conc u_k} = \DeadEnd$ if
$u_i$ is the beginning of an infinite jump chain.%
\footnote
{This rule can be formalized, cf.~\cite{BM07g}.}

\section{Interaction between Threads and Services}
\label{sect-TAbt}

A thread may perform a basic action for the purpose of requesting a
named service to process a method and to return a reply value at
completion of the processing of the method.
In this section, we extend \BTAbt\ such that this kind of interaction
between threads and services can be dealt with, resulting in \TAbt.
This involves the introduction of service families: collections of named
services.

It is assumed that a fixed but arbitrary set $\Meth$ of \emph{methods}
has been given.
Methods play the role of commands.
A service is able to process certain methods.
The processing of a method may involve a change of the service.
At completion of the processing of a method, the service produces a
reply value.
The set $\Replies$ of \emph{reply values} is the set
$\set{\True,\False,\Div}$.

In \SFA, the algebraic theory of service families introduced below, the
following is assumed with respect to services:
\begin{itemize}
\item
a set $\Services$ of services has been given together with:
\begin{itemize}
\item
for each $m \in \Meth$,
a total function $\funct{\effect{m}}{\Services}{\Services}$;
\item
for each $m \in \Meth$,
a total function $\funct{\sreply{m}}{\Services}{\Replies}$;
\end{itemize}
satisfying the condition that there exists a unique $S \in \Services$
with $\effect{m}(S) = S$ and $\sreply{m}(S) = \Div$ for all
$m \in \Meth$;
\item
a signature $\Sig{\Services}$ has been given that includes the following
sort:
\begin{itemize}
\item
the sort $\Serv$ of \emph{services};
\end{itemize}
and the following constant and operators:
\begin{itemize}
\item
the \emph{empty service} constant $\const{\emptyserv}{\Serv}$;
\item
for each $m \in \Meth$,
the \emph{derived service} operator
$\funct{\derive{m}}{\Serv}{\Serv}$;
\end{itemize}
\item
$\Services$ and $\Sig{\Services}$ are such that:
\begin{itemize}
\item
each service in $\Services$ can be denoted by a closed term of sort
$\Serv$;
\item
the constant $\emptyserv$ denotes the unique $S \in \Services$ such
that $\effect{m}(S) = S$ and $\sreply{m}(S) = \Div$ for all
$m \in \Meth$;
\item
if closed term $t$ denotes service $S$, then $\derive{m}(t)$ denotes
service $\effect{m}(S)$.
\end{itemize}
\end{itemize}

When a request is made to service $S$ to process method $m$:
\begin{itemize}
\item
if $\sreply{m}(S) \neq \Div$, then $S$ processes $m$, produces the reply
$\sreply{m}(S)$, and next proceeds as $\effect{m}(S)$;
\item
if $\sreply{m}(S) = \Div$, then $S$ rejects the request to process
method $m$.
\end{itemize}
The unique service $S$ such that $\effect{m}(S) = S$ and
$\sreply{m}(S) = \Div$ for all $m \in \Meth$ is called the \emph{empty}
service.
It is the service that is unable to process any method.

It is also assumed that a fixed but arbitrary non-empty finite set
$\Foci$ of \emph{foci} has been given.
Foci play the role of names of services in the service family offered by
an execution environment.
A service family is a set of named services where each name occurs only
once.

\SFA\ has the sorts, constants and operators in $\Sig{\Services}$
and in addition the following sort:
\begin{itemize}
\item
the sort $\ServFam$ of \emph{service families};
\end{itemize}
and the following constant and operators:
\begin{itemize}
\item
the \emph{empty service family} constant $\const{\emptysf}{\ServFam}$;
\item
for each $f \in \Foci$, the unary \emph{singleton service family}
operator $\funct{\mathop{f{.}} \ph}{\Serv}{\ServFam}$;
\item
the binary \emph{service family composition} operator
$\funct{\ph \sfcomp \ph}{\ServFam \x \ServFam}{\ServFam}$;
\item
for each $F \subseteq \Foci$, the unary \emph{encapsulation} operator
$\funct{\encap{F}}{\ServFam}{\ServFam}$.
\end{itemize}
We assume that there is a countably infinite set of variables of sort
$\ServFam$ which includes $u,v,w$.
Terms are built as usual in the many-sorted case
(see e.g.~\cite{ST99a,Wir90a}).\linebreak[2]
We use prefix notation for the singleton service family operators and
infix nota\-tion for the service family composition operator.

The service family denoted by $\emptysf$ is the empty service family.
The service family denoted by a closed term of the form $f.H$ consists
of one named service only, the service concerned is the service denoted
by $H$, and the name of this service is $f$.
The service family denoted by a closed term of the form $C \sfcomp D$
consists of all named services that belong to either the service family
denoted by $C$ or the service family denoted by $D$.
In the case where a named service from the service family denoted by $C$
and a named service from the service family denoted by $D$ have the same
name, they collapse to an empty service with the name concerned.
The service family denoted by a closed term of the form $\encap{F}(C)$
consists of all named services with a name not in $F$ that belong to the
service family denoted by $C$.

The axioms of \SFA\ are given in Table~\ref{axioms-SFA}.%
\begin{table}[!t]
\caption{Axioms of \SFA}
\label{axioms-SFA}
\begin{eqntbl}
\begin{axcol}
u \sfcomp \emptysf = u                                 & \axiom{SFC1} \\
u \sfcomp v = v \sfcomp u                              & \axiom{SFC2} \\
(u \sfcomp v) \sfcomp w = u \sfcomp (v \sfcomp w)      & \axiom{SFC3} \\
f.H \sfcomp f.H' = f.\emptyserv                        & \axiom{SFC4}
\end{axcol}
\qquad
\begin{saxcol}
\encap{F}(\emptysf) = \emptysf                       & & \axiom{SFE1} \\
\encap{F}(f.H) = \emptysf & \mif f \in F               & \axiom{SFE2} \\
\encap{F}(f.H) = f.H      & \mif f \notin F            & \axiom{SFE3} \\
\encap{F}(u \sfcomp v) =
\encap{F}(u) \sfcomp \encap{F}(v)                    & & \axiom{SFE4}
\end{saxcol}
\end{eqntbl}
\end{table}
In this table, $f$ stands for an arbitrary focus from $\Foci$ and $H$
and $H'$ stand for arbitrary closed terms of sort $\Serv$.
The axioms of \SFA\ simply formalize the informal explanation given above.

Below we will introduce two operators related to the interaction between
threads and services.
They are called the apply operator and the reply operator.
The apply operator is concerned with the effects of threads on service
families and therefore produces service families.
The reply operator is concerned with the effects of service families on
the Boolean values that threads deliver at their termination.
The reply operator does not only produce Boolean values: it produces a
special value in cases where no termination takes place.

For the set $\BAct$ of basic actions, we take the set
$\set{f.m \where f \in \Foci, m \in \Meth}$.
Both operators mentioned above relate to the processing of methods by
services from a service family in pursuance of basic actions performed
by a thread.
The service involved in the processing of a method is the service whose
name is the focus of the basic action in question.

\TAbt\ has the sorts, constants and operators of both \BTAbt\ and \SFA,
and in addition the following sort:
\begin{itemize}
\item
the sort $\Repl$ of \emph{replies};
\end{itemize}
and the following constants and operators:
\begin{itemize}
\item
the \emph{reply} constants $\const{\True,\False,\Div}{\Repl}$;
\item
the binary \emph{apply} operator
$\funct{\ph \sfapply \ph}{\Thr \x \ServFam}{\ServFam}$;
\item
the binary \emph{reply} operator
$\funct{\ph \sfreply \ph}{\Thr \x \ServFam}{\Repl}$.
\end{itemize}
We use infix notation for the apply and reply operators.

The service family denoted by a closed term of the form $p \sfapply C$
and the reply denoted by a closed term of the form $p \sfreply C$ are
the service family and reply, respectively, that result from processing
the method of each basic action with a focus of the service family
denoted by $C$ that the thread denoted by $p$ performs, where the
processing is done by the service in that service family with the focus
of the basic action as its name.
When the method of a basic action performed by a thread is processed by
a service, the service changes in accordance with the method concerned,
and affects the thread as follows: the two ways to proceed reduce to one
on the basis of the reply value produced by the service.
The reply is the Boolean value that the thread denoted by $p$ delivers
at termination if it terminates and the value $\Div$ (standing for
divergent) if it does not terminate.

The axioms of \TAbt\ are the axioms of \BTAbt, the axioms of \SFA, and
the axioms given in Tables~\ref{axioms-apply} and~\ref{axioms-reply}.%
\begin{table}[!t]
\caption{Axioms for apply operator}
\label{axioms-apply}
\begin{eqntbl}
\begin{saxcol}
\StopP \sfapply u = u                                  & & \axiom{A1} \\
\StopN \sfapply u = u                                  & & \axiom{A2} \\
\DeadEnd \sfapply u = \emptysf                         & & \axiom{A3} \\
(\Tau \bapf x) \sfapply u = x \sfapply u               & & \axiom{A4} \\
(\pcc{x}{f.m}{y}) \sfapply \encap{\set{f}}(u) = \emptysf
                                                       & & \axiom{A5} \\
(\pcc{x}{f.m}{y}) \sfapply (f.H \sfcomp \encap{\set{f}}(u)) =
x \sfapply (f.\derive{m}H \sfcomp \encap{\set{f}}(u))
                           & \mif \sreply{m}(H) = \True  & \axiom{A6} \\
(\pcc{x}{f.m}{y}) \sfapply (f.H \sfcomp \encap{\set{f}}(u)) =
y \sfapply (f.\derive{m}H \sfcomp \encap{\set{f}}(u))
                           & \mif \sreply{m}(H) = \False & \axiom{A7} \\
(\pcc{x}{f.m}{y}) \sfapply (f.H \sfcomp \encap{\set{f}}(u)) = \emptysf
                           & \mif \sreply{m}(H) = \Div   & \axiom{A8} \\
\AND{n \geq 0}{} \proj{n}{x} \sfapply u = \proj{n}{y} \sfapply v
                 \Implies x \sfapply u = y  \sfapply v & & \axiom{A9}
\end{saxcol}
\end{eqntbl}
\end{table}%
\begin{table}[!t]
\caption{Axioms for reply operator}
\label{axioms-reply}
\begin{eqntbl}
\begin{saxcol}
\StopP \sfreply u = \True                              & & \axiom{R1} \\
\StopN \sfreply u = \False                             & & \axiom{R2} \\
\DeadEnd \sfreply u = \Div                             & & \axiom{R3} \\
(\Tau \bapf x) \sfreply u = x \sfreply u               & & \axiom{R4} \\
(\pcc{x}{f.m}{y}) \sfreply \encap{\set{f}}(u) = \Div   & & \axiom{R5} \\
(\pcc{x}{f.m}{y}) \sfreply (f.H \sfcomp \encap{\set{f}}(u)) =
x \sfreply (f.\derive{m}H \sfcomp \encap{\set{f}}(u))
                           & \mif \sreply{m}(H) = \True  & \axiom{R6} \\
(\pcc{x}{f.m}{y}) \sfreply (f.H \sfcomp \encap{\set{f}}(u)) =
y \sfreply (f.\derive{m}H \sfcomp \encap{\set{f}}(u))
                           & \mif \sreply{m}(H) = \False & \axiom{R7} \\
(\pcc{x}{f.m}{y}) \sfreply (f.H \sfcomp \encap{\set{f}}(u)) = \Div
                           & \mif \sreply{m}(H) = \Div   & \axiom{R8} \\
\AND{n \geq 0}{} \proj{n}{x} \sfreply u = \proj{n}{y} \sfreply v
                 \Implies x \sfreply u = y  \sfreply v & & \axiom{R9}
\end{saxcol}
\end{eqntbl}
\end{table}
In these tables, $f$ stands for an arbitrary focus from $\Foci$, $m$
stands for an arbitrary method from $\Meth$, $H$ stands for an arbitrary
term of sort $\Serv$, and $n$ stands for an arbitrary natural number.
The axioms simply formalize the informal explanation given above and in
addition stipulate what is the result of apply and reply if
inappropriate foci or methods are involved.
Axioms A9 and R9 allow for reasoning about infinite threads in the
contexts of apply and reply, respectively.

Let $p$ and $C$ be \TAbt\ terms of sort $\Thr$ and $\ServFam$,
respectively.
Then $p$ \emph{converges on} $C$, written $p \cvg C$, is inductively
defined by the following clauses:
\begin{itemize}
\item
$\StopP \cvg u$ and $\StopN \cvg u$;
\item
if $x \cvg u$, then $(\Tau \bapf x) \cvg u$;
\item
if $\sreply{m}(H) = \True\hsp{-.025}$ and
$x \cvg (f.\derive{m}H \sfcomp \encap{\set{f}}(u))$, then
$(\pcc{x}{f.m}{y}) \cvg (f.H \sfcomp \encap{\set{f}}(u))$;
\item
if $\sreply{m}(H) = \False$ and
$y \cvg (f.\derive{m}H \sfcomp \encap{\set{f}}(u))$, then
$(\pcc{x}{f.m}{y}) \cvg (f.H \sfcomp \encap{\set{f}}(u))$;
\item
if $\proj{n}{x} \cvg u$, then $x \cvg u$.
\end{itemize}
Moreover, $p$ \emph{diverges on} $C$, written $p \dvg C$, is defined
by $p \dvg C$ iff not $p \cvg C$.

In the case where $p \dvg C$, either the processing of methods does
not halt or inappropriate foci or methods are involved.
In that case, there is nothing that we intend to denote by a term of the
form $p \sfapply C$ or $p \sfreply C$.
We propose to comply with the following \emph{relevant use conventions}:
\begin{itemize}
\item
$p \sfapply C$ is only used if it is known that $p \cvg C$;
\item
$p \sfreply C$ is only used if it is known that $p \cvg C$.
\end{itemize}
The condition found in the first convention is justified by the fact
that in the intended model of \TAbt, for definable threads $x$,
$x \sfapply u = \emptysf$ if $x \dvg u$ (see~\cite{BM09k}).
We do not have $x \sfapply u = \emptysf$ only if $x \dvg u$.
For instance, $\StopP \sfapply \emptysf = \emptysf$ whereas
$\StopP \cvg \emptysf$.

\section{Interaction between Programs and Services}
\label{sect-PGLBbt-add}

In this paper, the apply operator and reply operator are primarily
intended to be used in the setting of \PGLBbt.
In this section, we introduce the apply operator and reply operator in
the setting of \PGLBbt.
We also introduce notations for two simple transformations of \PGLBbt\
instruction sequences that affect only their termination behaviour on
execution and the Boolean value yielded at termination in the case of
termination.
These notations will be used in Sections~\ref{sect-interpreters}
and~\ref{sect-autosolvability}.

We introduce the apply operator and reply operator in the setting of
\PGAbt\ by defining:
\begin{ldispl}
x \sfapply u = \extr{x} \sfapply u\;, \quad
x \sfreply u = \extr{x} \sfreply u
\end{ldispl}
for all \PGLBbt\ instruction sequences $x$.
Similarly, we introduce convergence in the setting of \PGAbt\ by
defining:
\begin{ldispl}
x \cvg u = \extr{x} \cvg u
\end{ldispl}
for all \PGLBbt\ instruction sequences $x$.

The following proposition states that convergence corresponds with
termination.
\begin{proposition}
\label{prop-cvg-sfreply}
Let $x$ be a \PGLBbt\ instruction sequence.
Then $x \cvg u$ iff $x \sfreply u = \True$ or $x \sfreply u = \False$.
\end{proposition}
\begin{proof}
By the definition of $\extr{\ph}$, the last clause of the inductive
definition of $\cvg$, axiom R9, and Lemma~\ref{lemma-projections} it is
sufficient to prove $x \cvg u$ iff $p \sfreply u = \True$ or
$p \sfreply u = \False$ for each closed \BTAbt\ term $p$ of sort $\Thr$.
This is easy by induction on the structure of $p$.
\qed
\end{proof}

In Sections~\ref{sect-interpreters} and~\ref{sect-autosolvability}, we
will make use of two simple transformations of \PGLBbt\ instruction
sequences.
Here, we introduce notations for those transformations.

Let $x$ be a \PGLBbt\ instruction sequence.
Then we write $\swap(x)$ for $x$ with each occurrence of $\haltP$
replaced by $\haltN$ and each occurrence of $\haltN$ replaced by
$\haltP$, and we write $\ftod(x)$ for $x$ with each occurrence of
$\haltN$ replaced by $\fjmp{0}$.
In the following proposition, the most important properties relating to
these transformations are stated.
\begin{proposition}
\label{prop-swap-f2d}
Let $x$ be a \PGLBbt\ instruction sequence.
Then:
\begin{enumerate}
\item
if $x \sfreply u = \True$ then $\swap(x) \sfreply u = \False$ and
$\ftod(x) \sfreply u = \True$;
\item
if $x \sfreply u = \False$ then $\swap(x) \sfreply u = \True$ and
$\ftod(x) \sfreply u = \Div$.
\end{enumerate}
\end{proposition}
\begin{proof}
Let $p$ be a closed \BTAbt\ term of sort $\Thr$.
Then we write $\swap'(p)$ for $p$ with each occurrence of $\StopP$
replaced by $\StopN$ and each occurrence of $\StopN$ replaced by
$\StopP$, and we write $\ftod'(p)$ for $p$ with each occurrence of
$\StopN$ replaced by $\DeadEnd$.
It is easy to prove by induction on $i$ that
$\extr{i,\swap(x)} = \swap'(\extr{i,x})$ and
$\extr{i,\ftod(x)} = \ftod'(\extr{i,x})$ for all $i \in \Nat$.
By this result, axiom R9, and Lemma~\ref{lemma-projections} it is
sufficient to prove the following for each closed \BTAbt\ term $p$ of
sort $\Thr$:
\begin{enumerate}
\item[]
if $p \sfreply u = \True$ then $\swap'(p) \sfreply u = \False$ and
$\ftod'(p) \sfreply u = \True$;
\item[]
if $p \sfreply u = \False$ then $\swap'(p) \sfreply u = \True$ and
$\ftod'(p) \sfreply u = \Div$.
\end{enumerate}
This is easy by induction on the structure of $p$.
\qed
\end{proof}

\section{Functional Units}
\label{sect-func-unit}

In this section, we introduce the concept of a functional unit and
related concepts.
The concept of a functional unit was first introduced in~\cite{BM09l}.

It is assumed that a non-empty set $\FUS$ of \emph{states} has been
given.
As before, it is assumed that a non-empty finite set $\MN$ of methods
has been given.
However, in the setting of functional units, methods serve as names of
operations on a state space.
For that reason, the members of $\MN$ will henceforth be called
\emph{method names}.

A \emph{method operation} on $\FUS$ is a total function from $\FUS$ to
$\Bool \x \FUS$.
A \emph{partial method operation} on $\FUS$ is a partial function from
$\FUS$ to $\Bool \x \FUS$.
We write $\MO(\FUS)$ for the set of all method operations on $\FUS$.
We write $M^r$ and $M^e$, where $M \in \MO(\FUS)$, for the unique
functions $\funct{R}{\FUS}{\Bool}$ and $\funct{E}{\FUS}{\FUS}$,
respectively, such that $M(s) = \tup{R(s),E(s)}$ for all $s \in \FUS$.

A \emph{functional unit} for $\FUS$ is a finite subset $\cH$ of
$\MN \x \MO(\FUS)$ such that \mbox{$\tup{m,M} \in \cH$} and
$\tup{m,M'} \in \cH$ implies $M = M'$.
We write $\FU(\FUS)$ for the set of all functional units for $\FUS$.
We write $\IF(\cH)$, where $\cH \in \FU(\FUS)$, for the set
$\set{m \in \MN \where \Exists{M \in \MO(\FUS)}{\tup{m,M} \in \cH}}$.
We write $m_\cH$, where $\cH \in \FU(\FUS)$ and $m \in \IF(\cH)$, for
the unique $M \in \MO(\FUS)$ such that $\tup{m,M} \in \cH$.

We look upon the set $\IF(\cH)$, where $\cH \in \FU(\FUS)$, as the
interface of $\cH$.
It looks to be convenient to have a notation for the restriction of a
functional unit to a subset of its interface.
We write $\tup{I,\cH}$, where $\cH \in \FU(\FUS)$ and
$I \subseteq \IF(\cH)$, for the functional unit
$\set{\tup{m,M} \in \cH \where m \in I}$.

Let $\cH \in \FU(\FUS)$.
Then an \emph{extension} of $\cH$ is an $\cH' \in \FU(\FUS)$ such that
$\cH \subseteq \cH'$.

The following is a simple illustration of the use of functional units.
An unbounded counter can be modelled by a functional unit for $\Nat$
with method operations for set to zero, increment by one, decrement by
one, and test on zero.

According to the definition of a functional unit,
$\emptyset \in \FU(\FUS)$.
By that we have a unique functional unit with an empty interface, which
is not very interesting in itself.
However, when considering services that behave according to functional
units, $\emptyset$ is exactly the functional unit according to which the
empty service $\emptyserv$ (the service that is not able to process any
method) behaves.

The method names attached to method operations in functional units
should not be confused with the names used to denote specific method
operations in describing functional units.
Therefore, we will comply with the convention to use names beginning
with a lower-case letter in the former case and names beginning with an
upper-case letter in the latter case.

We will use \PGLBbt\ instruction sequences to derive partial method
operations from the method operations of a functional unit.
We write $\Lf{I}$, where $I \subseteq \MN$, for the set of all \PGLBbt\
instruction sequences, taking the set $\set{f.m \where m \in I}$ as the
set $\BInstr$ of basic instructions.

The derivation of partial method operations from the method operations
of a functional unit involves services whose processing of methods
amounts to replies and service changes according to corresponding method
operations of the functional unit concerned.
These services can be viewed as the behaviours of a machine, on which
the processing in question takes place, in its different states.
We take the set $\FU(\FUS) \x \FUS$ as the set $\Services$ of services.
We write $\cH(s)$, where $\cH \in \FU(\FUS)$ and $s \in \FUS$, for the
service $\tup{\cH,s}$.
The functions $\effect{m}$ and $\sreply{m}$ are defined as follows:
\begin{ldispl}
\begin{aeqns}
\effect{m}(\cH(s)) & = &
\Biggl\{
\begin{array}[c]{@{}l@{\;\;}l@{}}
\cH(m_\cH^e(s))            & \mif m \in \IF(\cH) \\
{\emptyset}(s')            & \mif m \notin \IF(\cH)\;,
\end{array}
\beqnsep
\sreply{m}(\cH(s))  & = &
\Biggl\{
\begin{array}[c]{@{}l@{\;\;}l@{}}
m_\cH^r(s) \phantom{\cH()} & \mif m \in \IF(\cH) \\
\Div                       & \mif m \notin \IF(\cH)\;,
\end{array}
\end{aeqns}
\end{ldispl}
where $s'$ is a fixed but arbitrary state in $S$.
We assume that each $\cH(s) \in \Services$ can be denoted by a closed
term of sort $\Serv$.
In this connection, we use the following notational convention: for each
$\cH(s) \in \Services$, we write $\cterm{\cH(s)}$ for an arbitrary
closed term of sort $\Thr$ that denotes $\cH(s)$.
The ambiguity thus introduced could be obviated by decorating $\cH(s)$
wherever it stands for a closed term.
However, in this paper, it is always immediately clear from the context
whether it stands for a closed term.
Moreover, we believe that the decorations are more often than not
distracting.
Therefore, we leave it to the reader to make the decorations mentally
wherever appropriate.

Let $\cH \in \FU(\FUS)$, and let $I \subseteq \IF(\cH)$.
Then an instruction sequence $x \in \Lf{I}$ produces a partial method
operation $\moextr{x}{\cH}$ as follows:
\begin{ldispl}
\begin{aceqns}
\moextr{x}{\cH}(s) & = &
\tup{\moextr{x}{\cH}^r(s),\moextr{x}{\cH}^e(s)}
 & \mif \moextr{x}{\cH}^r(s) = \True \Or
        \moextr{x}{\cH}^r(s) = \False\;, \\
\moextr{x}{\cH}(s) & \mathrm{is} & \mathrm{undefined}
 & \mif \moextr{x}{\cH}^r(s) = \Div\;,
\end{aceqns}
\end{ldispl}
where
\begin{ldispl}
\begin{aeqns}
\moextr{x}{\cH}^r(s) & = & x \sfreply f.\cterm{\cH(s)}\;, \\
\moextr{x}{\cH}^e(s) & = &
\mathrm{the\;unique}\; s' \in S\; \mathrm{such\;that}\;
 x \sfapply f.\cterm{\cH(s)} = f.\cterm{\cH(s')}\;.
\end{aeqns}
\end{ldispl}
If $\moextr{x}{\cH}$ is total, then it is called a
\emph{derived method operation} of $\cH$.

The binary relation $\below$ on $\FU(\FUS)$ is defined by
$\cH \below \cH'$ iff for all $\tup{m,M} \in \cH$, $M$ is a derived
method operation of $\cH'$.
The binary relation $\equiv$ on $\FU(\FUS)$ is defined by
$\cH \equiv \cH'$ iff $\cH \below \cH'$ and $\cH' \below \cH$.
In~\cite{BM09l}, it is proved that $\below$ is a quasi-order relation
and $\equiv$ is an equivalence relation.

\section{Functional Units Relating to Turing Machine Tapes}
\label{sect-func-unit-sbs}

In this section, we define some notions that have a bearing on the
halting problem in the setting of \PGLBbt\ and functional units.
The notions in question are defined in terms of functional units for the
following state space:
\begin{ldispl}
\SBS = \set{v \pebble w \where v,w \in \seqof{\set{0,1,\sep}}}\;.
\end{ldispl}

The states from $\SBS$ resemble the possible contents of the tape of a
Turing machine whose tape alphabet is $\set{0,1,\sep}$.
Consider a state $v \pebble w \in \SBS$.
Then $v$ corresponds to the content of the tape to the left of the
position of the tape head and $w$ corresponds to the content of the tape
from the position of the tape head to the right -- the indefinite
numbers of padding blanks at both ends are left out.
The colon serves as a seperator of bit sequences.
This is for instance useful if the input of a program consists of
another program and an input to the latter program, both encoded as a
bit sequences.

A method operation $M \in \MO(\SBS)$ is \emph{recursive} if there exist
recursive functions $\funct{F,G}{\Nat}{\Nat}$ such that
$M(v) = \tup{\beta(F(\alpha(v))),\alpha^{-1}(G(\alpha(v)))}$ for all
$v \in \SBS$,
where $\funct{\alpha}{\SBS}{\Nat}$ is a bijection and
$\funct{\beta}{\Nat}{\Bool}$ is inductively defined by $Z(0) = \True$
and $Z(n + 1) = \False$.
A functional unit $\cH \in \FU(\SBS)$ is \emph{recursive} if, for each
$\tup{m,M} \in \cH$, $M$ is recursive.

In the sequel, we will comply with the relevant use conventions
introduced at the end of Section~\ref{sect-TAbt}.

It is assumed that, for each $\cH \in \FU(\SBS)$, an injective function
from $\Lf{\IF(\cH)}$ to $\seqof{\set{0,1}}$ has been given that yields
for each $x \in \Lf{\IF(\cH)}$, an encoding of $x$ as a bit sequence.
We use the notation $\ol{x}$ to denote the encoding of $x$ as a bit
sequence.

Let $\cH \in \FU(\SBS)$, and let $I \subseteq \IF(\cH)$.
Then:
\begin{itemize}
\item
$x \in \Lf{\IF(\cH)}$ produces a
\emph{solution of the halting problem} for $\Lf{I}$ with respect to
$\cH$ if:
\begin{ldispl}
x \cvg f.\cterm{\cH(v)}\; \mathrm{for\; all}\; v \in \SBS\;, \\
x \sfreply f.\cterm{\cH(\pebble \ol{y} \sep v)} = \True \Iff
y \cvg f.\cterm{\cH(\pebble v)}\; \mathrm{for\; all}\;
y \in \Lf{I}\; \mathrm{and}\; v \in \seqof{\set{0,1,\sep}}\;;
\end{ldispl}
\item
$x \in \Lf{\IF(\cH)}$ produces a
\emph{reflexive solution of the halting problem} for $\Lf{I}$ with
respect to $\cH$ if $x$ produces a solution of the halting problem for
$\Lf{I}$ with respect to $\cH$ and $x \in \Lf{I}$;
\item
the halting problem for $\Lf{I}$ with respect to $\cH$ is
\emph{autosolvable} if there exists an $x \in \Lf{\IF(\cH)}$ such that
$x$ produces a reflexive solution of the halting problem for $\Lf{I}$
with respect to $\cH$;
\item
the halting problem for $\Lf{I}$ with respect to $\cH$ is
\emph{potentially autosolvable} if there exist an extension $\cH'$ of
$\cH$ and the halting problem for $\Lf{\IF(\cH')}$ with respect to
$\cH'$ is autosolvable;
\item
the halting problem for $\Lf{I}$ with respect to $\cH$ is
\emph{potentially recursively autosolvable} if there exist an extension
$\cH'$ of $\cH$ and the halting problem for $\Lf{\IF(\cH')}$ with
respect to $\cH'$ is autosolvable and $\cH'$ is recursive.
\end{itemize}
These definitions make clear that each combination of an
$\cH \in \FU(\SBS)$ and an $I \subseteq \IF(\cH)$ gives rise to a
\emph{halting problem instance}.

In Section~\ref{sect-interpreters} and~\ref{sect-autosolvability}, we
will make use of a method operation $\Dup \in \MO(\SBS)$ for duplicating
bit sequences.
This method operation is defined as follows:
\begin{ldispl}
\begin{aceqns}
\Dup(v \pebble w) & = & \Dup(\pebble v w)\;, \\
\Dup(\pebble v)   & = & \tup{\True,\pebble v \sep v}
 & \mif v \in \seqof{\set{0,1}}\;, \\
\Dup(\pebble v \sep w) & = & \tup{\True,\pebble v \sep v \sep w}
 & \mif v \in \seqof{\set{0,1}}\;.
\end{aceqns}
\end{ldispl}

\begin{proposition}
\label{prop-dup}
Let $\cH \in \FU(\SBS)$ be such that $\tup{\dup,\Dup} \in \cH$,
let $I \subseteq  \IF(\cH)$ be such that $\dup \in I$,
let $x \in \Lf{I}$, and
let $v \in \seqof{\set{0,1}}$ and $w \in \seqof{\set{0,1,\sep}}$ be such
that $w = v$ or $w = v \sep w'$ for some $w' \in \seqof{\set{0,1,\sep}}$.
Then
$(f.\dup \conc x) \sfreply f.\cterm{\cH(\pebble w)} =
 x \sfreply f.\cterm{\cH(\pebble v \sep w)}$.
\end{proposition}
\begin{proof}
This follows immediately from the definition of $\Dup$ and the axioms
for~$\sfreply$.
\qed
\end{proof}

By the use of foci and the introduction of apply and reply operators on
service families, we make it possible to deal with cases that remind of
multi-tape Turing machines, Turing machines that has random access
memory, etc.
However, in this paper, we will only consider the case that reminds of
single-tape Turing machines.
This means that we will use only one focus ($f$) and only
singleton service families.

\section{Interpreters}
\label{sect-interpreters}

It is often mentioned that an interpreter, which is a program for
simulating the execution of programs that it is given as input, cannot
solve the halting problem because the execution of the interpreter will
not terminate if the execution of its input program does not terminate.
In this section, we have a look upon the termination behaviour of
interpreters in the setting of \PGLBbt\ and functional units.

Let $\cH \in \FU(\SBS)$, let $I \subseteq \IF(\cH)$, and
let $I' \subseteq I$.
Then $x \in \Lf{I}$ is an \emph{interpreter} for $\Lf{I'}$ with respect
to $\cH$ if for all $y \in \Lf{I'}$ and $v \in \seqof{\set{0,1,\sep}}$:
\begin{ldispl}
y \cvg f.\cterm{\cH(\pebble v)} \Implies
x \cvg f.\cterm{\cH(\pebble \ol{y} \sep v)}\;, \\
x \sfapply f.\cterm{\cH(\pebble \ol{y} \sep v)} =
y \sfapply f.\cterm{\cH(\pebble v)}\; \mathrm{and}\;
x \sfreply f.\cterm{\cH(\pebble \ol{y} \sep v)} =
y \sfreply f.\cterm{\cH(\pebble v)}\;.
\end{ldispl}
Moreover, $x \in \Lf{I}$ is a \emph{reflexive interpreter} for $\Lf{I'}$
with respect to $\cH$ if $x$ is an interpreter for $\Lf{I'}$ with respect
to $\cH$ and $x \in \Lf{I'}$.

The following theorem states that a reflexive interpreter that always
terminates is impossible in the presence of the method operation $\Dup$.
\begin{theorem}
\label{theorem-interpreter}
Let $\cH \in \FU(\SBS)$ be such that $\tup{\dup,\Dup} \in \cH$,
let $I \subseteq \IF(\cH)$ be such that $\dup \in I$, and
let $x \in \Lf{\IF(\cH)}$ be a reflexive interpreter for $\Lf{I}$ with
respect to $\cH$.
Then there exist an $y \in \Lf{I}$ and a $v \in \seqof{\set{0,1,\sep}}$
such that $x \dvg f.\cterm{\cH(\pebble \ol{y} \sep v)}$.
\end{theorem}
\begin{proof}
Assume the contrary.
Take $y = f.\dup \conc \swap(x)$.
By the assumption, $x \cvg f.\cterm{\cH(\pebble \ol{y} \sep \ol{y})}$.
By Propositions~\ref{prop-cvg-sfreply} and~\ref{prop-swap-f2d}, it
follows that $\swap(x) \cvg f.\cterm{\cH(\pebble \ol{y} \sep \ol{y})}$
and
$\swap(x) \sfreply f.\cterm{\cH(\pebble \ol{y} \sep \ol{y})} \neq
 x \sfreply f.\cterm{\cH(\pebble \ol{y} \sep \ol{y})}$.
By Propositions~\ref{prop-cvg-sfreply} and~\ref{prop-dup}, it follows
that $(f.\dup \conc \swap(x)) \cvg f.\cterm{\cH(\pebble \ol{y})}$ and
$(f.\dup \conc \swap(x)) \sfreply f.\cterm{\cH(\pebble \ol{y})} \neq
 x \sfreply f.\cterm{\cH(\pebble \ol{y} \sep \ol{y})}$.
Since $y = f.\dup \conc \swap(x)$, we have
$y \cvg f.\cterm{\cH(\pebble \ol{y})}$ and
$y \sfreply f.\cterm{\cH(\pebble \ol{y})} \neq
 x \sfreply f.\cterm{\cH(\pebble \ol{y} \sep \ol{y})}$.
Because $x$ is a reflexive interpreter, this implies
$x \sfreply f.\cterm{\cH(\pebble \ol{y} \sep \ol{y})} =
 y \sfreply f.\cterm{\cH(\pebble \ol{y})}$ and
$y \sfreply f.\cterm{\cH(\pebble \ol{y})} \neq
 x \sfreply f.\cterm{\cH(\pebble \ol{y} \sep \ol{y})}$.
This is a contradiction.
\qed
\end{proof}
In the proof of Theorem~\ref{theorem-interpreter}, the presence of the
method operation $\Dup$ is essential.
It is easy to see that the theorem goes through for all functional units
for $\SBS$ of which $\Dup$ is a derived method operation.
An example of such a functional unit is the one whose method operations
correspond to the basic steps that can be performed on the tape of a
Turing machine.

For each $\cH \in \FU(\SBS)$, $m \in \IF(\cH)$, and $v \in \SBS$,
we have $(f.m \conc \haltP \conc \haltN) \cvg f.\cterm{\cH(v)}$.
This leads us to the following corollary of
Theorem~\ref{theorem-interpreter}.
\begin{corollary}
\label{corollary-interpreter}
For all $\cH \in \FU(\SBS)$ with $\tup{\dup,\Dup} \in \cH$ and
$I \subseteq \IF(\cH)$ with $\dup \in I$, there does not exist an
$m \in I$ such that $f.m \conc \haltP \conc \haltN$ is a reflexive
interpreter for $\Lf{I}$ with respect to $\cH$.
\end{corollary}

\section{Autosolvability of the Halting Problem}
\label{sect-autosolvability}

Because a reflexive interpreter that always terminates is impossible in
the presence of the method operation $\Dup$, we must conclude that
solving the halting problem by means of a reflexive interpreter is out
of the question in the presence of the method operation $\Dup$.
The question arises whether the proviso ``by means of a reflexive
interpreter'' can be dropped.
In this section, we answer this question in the affirmative.
Before we present this negative result concerning autosolvability of the
halting problem, we present a positive result.

Let $M \in \MO(\SBS)$.
Then we say that $M$ \emph{increases the number of colons} if for some
$v \in \SBS$ the number of colons in $M^e(v)$ is greater than the number
of colons in $v$.

\begin{theorem}
\label{theorem-autosolv}
Let $\cH \in \FU(\SBS)$ be such that no method operation of $\cH$
increases the number of colons.
Then there exist an extension $\cH'$ of $\cH$,
an $I' \subseteq \IF(\cH')$, and an $x \in \Lf{\IF(\cH')}$ such that
$x$ produces a reflexive solution of the halting problem for $\Lf{I'}$
with respect to $\cH'$.
\end{theorem}
\begin{proof}
Let $\halting \in \MN$ be such that $\halting \notin \IF(\cH)$.
Take $I' = \IF(\cH) \union \set{\halting}$.
Take $\cH' = \cH \union \set{\tup{\halting,\Halting}}$, where
$\Halting \in \MO(\SBS)$ is defined by induction on the number of
colons in the argument of $\Halting$ as follows:
\begin{ldispl}
\begin{aceqns}
\Halting(v \pebble w) & = & \Halting(\pebble v w)\;, \\
\Halting(\pebble v)   & = & \tup{\False,\pebble}
 & \mif v \in \seqof{\set{0,1}}\;, \\
\Halting(\pebble v \sep w) & = & \tup{\False,\pebble}
 & \mif v \in \seqof{\set{0,1}} \And
        \Forall{x \in \Lf{I'}}{v \neq \ol{x}}\;, \\
\Halting(\pebble \ol{x} \sep w) & = & \tup{\False,\pebble}
 & \mif x \in \Lf{I'} \And x \dvg f.\cterm{\cH'(w)}\;, \\
\Halting(\pebble \ol{x} \sep w) & = & \tup{\True,\pebble}
 & \mif x \in \Lf{I'} \And x \cvg f.\cterm{\cH'(w)}\;.
\end{aceqns}
\end{ldispl}
Then $\ptst{f.\halting} \conc \haltP \conc \haltN$ produces a reflexive
solution of the halting problem for $\Lf{I'}$ with respect to $\cH'$.
\qed
\end{proof}
Theorem~\ref{theorem-autosolv} tells us that there exist functional
units $\cH \in \FU(\SBS)$ with the property that the halting problem is
potentially autosolvable for $\Lf{\IF(\cH)}$ with respect to $\cH$.
Thus, we know that there exist functional units $\cH \in \FU(\SBS)$ with
the property that the halting problem is autosolvable for
$\Lf{\IF(\cH)}$ with respect to~$\cH$.

There exists an $\cH \in \FU(\SBS)$ for which $\Halting$ as defined in
the proof of Theorem~\ref{theorem-autosolv} is computable, and hence
recursive.
\begin{theorem}
\label{theorem-comput}
Let $\cH = \emptyset$ and
$\cH' = \cH \union \set{\tup{\halting,\Halting}}$, where
$\Halting$ is as defined in the proof
of Theorem~\ref{theorem-autosolv}.
Then, $\Halting$ is computable.
\end{theorem}
\begin{proof}
It is sufficient to prove for an arbitrary $x \in \Lf{\IF(\cH')}$ that,
for all $v \in \SBS$, $x \cvg f.\cterm{\cH'(v)}$ is decidable.
We will prove this by induction on the number of colons in $v$.

The basis step.
Because the number of colons in $v$ equals $0$,
$\Halting(v) = \tup{\False,\pebble}$.
It follows that $x \cvg f.\cterm{\cH'(v)} \Iff x' \cvg \emptysf$,
where $x'$ is $x$ with each occurrence of $f.\halting$ and
$\ptst{f.\halting}$ replaced by $\fjmp{2}$ and each occurrence of
$\ntst{f.\halting}$ replaced by $\fjmp{1}$.
Because $x'$ is finite, $x' \cvg \emptysf$ is decidable.
Hence, $x \cvg f.\cterm{\cH'(v)}$ is decidable.

The inductive step.
Because the number of colons in $v$ is greater than $0$, either
$\Halting(v) = \tup{\True,\pebble}$ or
$\Halting(v) = \tup{\False,\pebble}$.
It follows that $x \cvg f.\cterm{\cH'(v)} \Iff x' \cvg \emptysf$, where
$x'$ is $x$ with:
\begin{itemize}
\item
each occurrence of $f.\halting$ and $\ptst{f.\halting}$ replaced by
$\fjmp{1}$ if the occurrence leads to the first application of
$\Halting$ and $\Halting^r(v) = \True$, and by $\fjmp{2}$ otherwise;
\item
each occurrence of $\ntst{f.\halting}$ replaced by
$\fjmp{2}$ if the occurrence leads to the first application of
$\Halting$ and $\Halting^r(v) = \True$, and by $\fjmp{1}$ otherwise.
\end{itemize}
An occurrence of $f.\halting$, $\ptst{f.\halting}$ or
$\ntst{f.\halting}$ in $x$ leads to the first application of $\Halting$
iff $\extr{1,x} = \extr{i,x}$, where $i$ is its position in $x$.
Because $x$ is finite, it is decidable whether an occurrence of
$f.\halting$, $\ptst{f.\halting}$ or $\ntst{f.\halting}$ leads to the
first processing of $\halting$.
Moreover, by the induction hypothesis, it is decidable whether
$\Halting^r(v) = \True$.
Because $x'$ is finite, it follows that $x' \cvg \emptysf$ is decidable.
Hence, $x \cvg f.\cterm{\cH'(v)}$ is decidable.
\qed
\end{proof}
Theorems~\ref{theorem-autosolv} and~\ref{theorem-comput} together tell
us that there exists a functional unit $\cH \in \FU(\SBS)$, viz.\
$\emptyset$, with the property that the halting problem is potentially
recursively autosolvable for $\Lf{\IF(\cH)}$ with respect to $\cH$.

There exist functional units in $\FU(\SBS)$ of which all recursive
$M \in \MO(\SBS)$ that do not increase the number of colons are derived
method operations.
A witness is the functional unit whose method operations correspond to
the basic steps that can be performed on the tape of a Turing machine.
Let $\cH \in \FU(\SBS)$ be such that all recursive $M \in \MO(\SBS)$
that do not increase the number of colons are derived method operations
of $\cH$.
Then the halting problem is potentially autosolvable for $\Lf{\IF(\cH)}$
with respect to $\cH$.
However, the halting problem is not potentially recursively autosolvable
for $\Lf{\IF(\cH)}$ with respect to $\cH$ because otherwise the halting
problem would be decidable.

The following theorem tells us essentially that potential
autosolvability of the halting problem is precluded in the presence of
the method operation $\Dup$.
\begin{theorem}
\label{theorem-non-autosolv}
Let $\cH \in \FU(\SBS)$ be such that $\tup{\dup,\Dup} \in \cH$, and
let $I \subseteq \IF(\cH)$ be such that $\dup \in I$.
Then there does not exist an $x \in \Lf{\IF(\cH)}$ such that $x$
produces a reflexive solution of the halting problem for $\Lf{I}$ with
respect to $\cH$.
\end{theorem}
\begin{proof}
Assume the contrary.
Let $x \in \Lf{\IF(\cH)}$ be such that $x$ produces a re\-flexive
solution of the halting problem for $\Lf{I}$ with respect to $\cH$, and
let $y = f.\dup \conc \ftod(\swap(x))$.
Then $x \cvg f.\cterm{\cH(\pebble \ol{y} \sep \ol{y})}$.
By Propositions~\ref{prop-cvg-sfreply} and~\ref{prop-swap-f2d}, it
fol\-lows that $\swap(x) \cvg f.\cterm{\cH(\pebble \ol{y} \sep \ol{y})}$
and either
$\swap(x) \sfreply f.\cterm{\cH(\pebble \ol{y} \sep \ol{y})} = \True$ or
$\swap(x) \sfreply\linebreak[2] f.\cterm{\cH(\pebble \ol{y} \sep \ol{y})}
  = \False$.

In the case where
$\swap(x) \sfreply f.\cterm{\cH(\pebble \ol{y} \sep \ol{y})} = \True$,
we have by Proposition~\ref{prop-swap-f2d} that
(i)~$\ftod(\swap(x)) \sfreply f.\cterm{\cH(\pebble \ol{y} \sep \ol{y})}
       = \True$ and
(ii)~$x \sfreply f.\cterm{\cH(\pebble \ol{y} \sep \ol{y})} = \False$.
By Proposition~\ref{prop-dup}, it follows from~(i) that
$(f.\dup \conc \ftod(\swap(x))) \sfreply f.\cterm{\cH(\pebble \ol{y})} =
 \True$.
Since $y = f.\dup \conc\linebreak \ftod(\swap(x))$, we have
$y \sfreply f.\cterm{\cH(\pebble \ol{y})} = \True$.
On the other hand, because $x$ produces a reflexive solution, it follows
from~(ii) that $y \dvg f.\cterm{\cH(\pebble \ol{y})}$.
By Proposition~\ref{prop-cvg-sfreply}, this contradicts with
$y \sfreply f.\cterm{\cH(\pebble \ol{y})} = \True$.

In the case where
$\swap(x) \sfreply f.\cterm{\cH(\pebble \ol{y} \sep \ol{y})} = \False$,
we have by Proposition~\ref{prop-swap-f2d} that
(i)~$\ftod(\swap(x)) \sfreply f.\cterm{\cH(\pebble \ol{y} \sep \ol{y})}
       = \Div$ and
(ii)~$x \sfreply f.\cterm{\cH(\pebble \ol{y} \sep \ol{y})} =
\True$.
By Proposition~\ref{prop-dup}, it follows from~(i) that
$(f.\dup \conc \ftod(\swap(x))) \sfreply f.\cterm{\cH(\pebble \ol{y})} =
 \Div$.
Since $y = f.\dup \conc \ftod(\swap(x))$, we have
$y \sfreply f.\cterm{\cH(\pebble \ol{y})} = \Div$.
On the other hand, because $x$ produces a reflexive solution, it follows
from~(ii) that $y \cvg f.\cterm{\cH(\pebble \ol{y})}$.
By Proposition~\ref{prop-cvg-sfreply}, this contradicts with
$y \sfreply f.\cterm{\cH(\pebble \ol{y})} = \Div$.
\qed
\end{proof}

Below, we will give an alternative proof of
Theorem~\ref{theorem-non-autosolv}.
A case distinction is needed in both proofs, but in the alternative
proof it concerns a minor issue.
The issue in question is covered by the following lemma.
\begin{lemma}
\label{lemma-non-autosolv}
Let $\cH \in \FU(\SBS)$, let $I \subseteq \IF(\cH)$,
let $x \in \Lf{\IF(\cH)}$ be such that $x$ produces a reflexive solution
of the halting problem for $\Lf{I}$ with respect to $\cH$,
let $y \in \Lf{I}$, and let $v \in \seqof{\set{0,1,\sep}}$.
Then
$y \cvg f.\cterm{\cH(\pebble v)}$ implies
$y \sfreply f.\cterm{\cH(\pebble v)} =
 x \sfreply f.\cterm{\cH(\pebble \ol{\ftod(y)} \sep v)}$.
\end{lemma}
\begin{proof}
By Proposition~\ref{prop-cvg-sfreply}, it follows from
$y \cvg f.\cterm{\cH(\pebble v)}$ that either
$y \sfreply f.\cterm{\cH(\pebble v)} = \True$ or
$y \sfreply f.\cterm{\cH(\pebble v)} = \False$.

In the case where $y \sfreply f.\cterm{\cH(\pebble v)} = \True$, we have
by Propositions~\ref{prop-cvg-sfreply} and~\ref{prop-swap-f2d} that
$\ftod(y) \cvg f.\cterm{\cH(\pebble v)}$ and so
$x \sfreply f.\cterm{\cH(\pebble \ol{\ftod(y)} \sep v)} = \True$.

In the case where $y \sfreply f.\cterm{\cH(\pebble v)} = \False$, we have
by Propositions~\ref{prop-cvg-sfreply} and~\ref{prop-swap-f2d} that
$\ftod(y) \dvg f.\cterm{\cH(\pebble v)}$ and so
$x \sfreply f.\cterm{\cH(\pebble \ol{\ftod(y)} \sep v)} = \False$.
\qed
\end{proof}

\begin{trivlist}
\item
\emph{Another proof of Theorem~\ref{theorem-non-autosolv}.}
Assume the contrary.
Let $x \in \Lf{\IF(\cH)}$ be such that $x$ produces a reflexive solution
of the halting problem for $\Lf{I}$ with re\-spect to $\cH$, and
let $y = \ftod(\swap(f.\dup \conc x))$.
Then $x \cvg f.\cterm{\cH(\pebble \ol{y} \sep \ol{y})}$.
By Propo\-sitions~\ref{prop-cvg-sfreply}, \ref{prop-swap-f2d}
and~\ref{prop-dup}, it follows that
$\swap(f.\dup \conc x) \cvg f.\cterm{\cH(\pebble \ol{y})}$.
By Lemma~\ref{lemma-non-autosolv}, it follows that
$\swap(f.\dup \conc x) \sfreply f.\cterm{\cH(\pebble \ol{y})} =
  x \sfreply f.\cterm{\cH(\pebble \ol{y} \sep \ol{y})}$.
By Proposition~\ref{prop-swap-f2d}, it follows that
$(f.\dup \conc x) \sfreply f.\cterm{\cH(\pebble \ol{y})} \neq
 x \sfreply f.\cterm{\cH(\pebble \ol{y} \sep \ol{y})}$.
On the other hand, by Propo\-sition~\ref{prop-dup}, we have that
$(f.\dup \conc x) \sfreply f.\cterm{\cH(\pebble \ol{y})} =
 x \sfreply f.\cterm{\cH(\pebble \ol{y} \sep \ol{y})}$.
This contradicts with
$(f.\dup \conc x) \sfreply f.\cterm{\cH(\pebble \ol{y})} \neq
 x \sfreply f.\cterm{\cH(\pebble \ol{y} \sep \ol{y})}$.
\qed
\end{trivlist}

Let $\cH = \set{\tup{\dup,\Dup}}$.
By Theorem~\ref{theorem-non-autosolv}, the halting problem for
$\Lf{\set{\dup}}$ with respect to $\cH$ is not (potentially)
autosolvable.
However, it is decidable.
\begin{theorem}
\label{theorem-decidable}
Let $\cH = \set{\tup{\dup,\Dup}}$.
Then the halting problem for $\Lf{\set{\dup}}$ with respect to $\cH$ is
decidable.
\end{theorem}
\begin{proof}
Let $x \in \Lf{\set{\dup}}$, and let $x'$ be $x$ with each occurrence of
$f.\dup$ and $\ptst{f.\dup}$ replaced by $\fjmp{1}$ and each occurrence
of $\ntst{f.\dup}$ replaced by $\fjmp{2}$.
For all $v \in \SBS$, $\Dup^r(v) = \True$.
Therefore, $x \cvg f.\cH(v) \Iff x' \cvg \emptysf$ for all $v \in \SBS$.
Because $x'$ is finite, $x' \cvg \emptysf$ is decidable.
\qed
\end{proof}
Both proofs of Theorem~\ref{theorem-non-autosolv} given above are
diagonalization proofs in disguise.
Theorem~\ref{theorem-decidable} indicates that diagonalization and
decidability are independent so to speak.

\section{Concluding Remarks}
\label{sect-concl}

We have extended and strengthened the results regarding the halting
problem for programs given in~\cite{BP04a} in a setting which looks to
be more adequate to describe and analyse issues regarding the halting
problem for programs.

It happens that decidability depends on the halting problem instance
considered.
This is different in the case of the on-line halting problem for
programs, i.e.\ the problem to forecast during its execution whether a
program will eventually terminate (see~\cite{BP04a}).

An interesting option for future work is to investigate the bounded
halting problem for programs, i.e.\ the problem to determine, given a
program and an input to the program, whether execution of the program on
that input will terminate after the execution of no more than a fixed
number of basic instructions.

\bibliographystyle{splncs03}
\bibliography{TA}


\end{document}